\documentclass[titlepage,11pt]{article}
\title{Detecting an induced net subdivision}
\author{Maria Chudnovsky\thanks{Partially supported by
    NSF grants DMS-0758364 and DMS-1001091.}\\
Columbia University, New York, NY 10027
\\
\\
Paul Seymour\thanks{Partially supported by ONR grant N00014-01-1-0608 and NSF grant DMS-0070912.}\\
Princeton University, Princeton, NJ 08544
\\
\\
Nicolas Trotignon\thanks{Partially supported by the French \emph{Agence Nationale de la
Recherche} under reference \textsc{anr-10-jcjc-Heredia.}}\\
CNRS, LIP -- ENS Lyon\\15 parvis Ren\'e Descartes, BP 7000, 69342 Lyon cedex 07, France}
\date{June 20, 2011; revised July 12, 2013}

\def\d{\hbox{-}}
\def\c{\hbox{-}\cdots\hbox{-}}

\usepackage{amsthm}
\usepackage{latexsym}
\usepackage{amsfonts}
\usepackage{amssymb}
\usepackage{amsbsy}
\usepackage{mathrsfs} 
\newtheorem{theorem}{}[section] 
\newtheorem{lemma}[theorem]{}
\def\l{,\ldots,}
\newcommand{\sm}{\setminus} %prive de 

 \renewenvironment{proof}[1][]%
 {\noindent {\setcounter{claim}{0}\it Proof.
    }{#1}{}}{\hfill$\Box$\vspace{2ex}} 

\newcounter{claim}
\begin{document}
\maketitle 

\begin{abstract}
A {\em net} is a graph consisting of a triangle $C$ and three more vertices, each of degree one and with 
its neighbour in $C$, and all adjacent to different vertices of $C$. We give a polynomial-time algorithm
to test whether an input graph has an induced subgraph which is a subdivision of a net. Unlike many similar questions, this does not
seem to be solvable by an application of the ``three-in-a-tree'' subroutine.
\end{abstract}

\section{Introduction} 

In this paper, all graphs are simple and finite.  
Let $H$ be the graph with six vertices $a_1, a_2, a_3, b_1, b_2 , b_3$
and the following edges: $a_1a_2, a_2a_3, a_3a_1, b_1a_1, b_2a_2,
b_3a_3$.  This is called a {\em net}. A \emph{doily} is a graph consisting of a cycle $H$ and three more vertices $b_1,b_2,b_3$, pairwise non-adjacent,
such that each $b_i$ has a unique neighbour $a_i$ in $V(H)$, and $a_1,a_2,a_3$ are all different.
Thus some induced subgraph of a graph $G$ is a doily if and only if some induced subgraph of $G$ is a subdivision of a net.
A \emph{doily
  of a graph $G$} is an induced subgraph of $G$ that is a doily.
We say a graph \emph{contains} a doily if some induced subgraph is a doily.

In this paper, we give a polynomial-time algorithm to test whether an input graph $G$ contains
a doily. Before we go on, let us motivate this a little. For any fixed graph $H$, we can test
if an input graph $G$ contains $H$ as a subgraph, or as an induced subgraph, in time $O(n^{|V(H)|})$, just by checking all
sets of vertices of $G$ of cardinality $|V(H)|$. (When we give the running time
of an algorithm whose input is a graph $G$, $n$ stands for the number
of vertices of $G$.) And these count as polynomial-time algorithms, if $H$ is a fixed graph.
One can also check in polynomial time (again, with $H$ fixed) whether some subgraph of $G$ is a subdivision of $H$,
as a consequence of the results of the Graph Minors series~\cite{RS13}.

On the other hand, checking whether $G$ contains an {\em induced} subgraph isomorphic to a subdivision of $H$
is much more complicated. Let us call this the ``induced $H$ subdivision problem''. 
For some graphs $H$ this can be solved in polynomial time, and for some it is NP-complete,
and we are far from identifying the border between the two. For instance, the following seem to be open:
\begin{itemize}
\item Can it be solved in polynomial time for every graph $H$ with maximum degree at most three?
\item Can it be solved in polynomial time when $H$ is $K_4$?
\item Can it be solved in polynomial time when $H$ consists of two disjoint triangles?
\end{itemize}

Here are some results: 

\begin{theorem}\label{NPcompletecases}
The induced $H$ subdivision problem is NP-complete when $H$ is either
\begin{itemize}
\item the graph obtained from the complete bipartite graph $K_{2,3}$ by adding an edge joining the two vertices of degree three, or
\item the graph with seven vertices $1\l 7$ and edges $12,13,23,14,15,26,27$, or
\item the tree obtained by adding nine leaves to a three-vertex path $P$, three adjacent to each vertex of $P$.
\end{itemize}
It is polynomial-time solvable when $H$ is either
\begin{itemize}
\item the complete bipartite graph $K_{2,3}$, or
\item the graph with six vertices $1\l 6$ and edges $12,13,23,14,15,26$, or
\item a tree $H$ that can be obtained as follows: let $T_1$ be a tree with at most four vertices, let $T_2$ be obtained from $T_1$
by adding arbitrarily many leaves each adjacent to some leaf of $T_1$, and let $H$ be a subdivision of $T_2$.
\end{itemize}
\end{theorem}
The first three results are consequences of a result of~\cite{leveque.lmt:detect}. 
The last three are all by application of an algorithm of~\cite{chudnovsky.seymour:theta}, that we call the ``three-in-a-tree'' algorithm; given
a graph $G$ and three vertices of $G$, it tests if there is a subset of $V(G)$ inducing a tree that contains the three vertices.
The three-in-a-tree algorithm was given several more applications in~\cite{leveque.lmt:detect}. Indeed, to date all the non-trivial polynomial-time
instances of the induced $H$ subdivision problem were solved by application of the three-in-a-tree algorithm. The result of this
paper is {\em not} such an application. Our main result is that when $H$ is a net, the induced $H$ subdivision problem is polynomial-time solvable. 
More explicitly:

\begin{theorem}
  \label{th:main}
  There is an $O(n^{16})$-time algorithm whose input is a graph $G$ and
  whose output is a doily of $G$ if such a doily exists.
\end{theorem}

In fact, one can easily modify the algorithm we present so that it outputs a doily of $G$ with the minimum number of vertices,
when one exists, but we omit those details.

\section{Outline of the algorithm}
\label{sec:outline}

A \emph{frame} in a graph $G$ is a twelve-tuple 
$$(b_1, b_2, b_3,
a_1, a_2, a_3, a'_1, a_2', a_3', a_1'', a_2'', a_3'')$$
of vertices from $V(G)$ such that 
\begin{itemize}
\item $a_1$,
$a_2$, $a_3$, $b_1$, $b_2$, and  $b_3$ are distinct, 
\item for $i = 1,2,3$, $b_i$ has degree one and $a_ib_i\in E(G)$, and
\item for $i = 1,2,3$, $a_i$ has degree three and its neighbours are $b_i,a_i',a_i''$.
\end{itemize}

Let $K$ be a doily of some graph $G$.  (In what follows, all computations with indices are
modulo three). A frame 
$$(b_1, b_2, b_3,
a_1, a_2, a_3, a'_1, a_2', a_3', a_1'', a_2'', a_3'')$$
in $G$
is a \emph{frame for $K$} if:

\begin{itemize}
\item $b_1,b_2,b_3\in V(K)$, and therefore $a_1, a_2, a_3, a'_1, a_2', a_3', a_1'', a_2'', a_3'' \in V(K)$
\item for $i = 1,2,3$, $a''_{i+1}$ and $a'_{i-1}$ belong to the path of $K$ between $a_{i-1}$ and $a_{i+1}$ not containing $a_i$.
\end{itemize}

If $v$ is a vertex of a graph $G$, $N(v)$ denotes the set of neighbours of $v$ in $G$.
Let $G$ be a graph and $X, Y \subseteq V(G)$ be disjoint non-empty
sets.  A vertex $v\in V(G)\sm (X\cup Y)$ is \emph{the centre of a star
  cutset that separates $X$ from $Y$} if for some $S\subseteq N(v) \sm
(X\cup Y)$, the graph $G\sm (\{v\} \cup S)$ contains no path from $X$
to $Y$ (in particular, $G\sm (\{v\} \cup S)$ is disconnected).

A doily of a graph $G$ is \emph{minimum} if its number of vertices is
minimum over all doilies of $G$. A pair $(G, F)$, where $G$ is a
graph and $F = (b_1, b_2, b_3, a_1,a_2,a_3 \dots)$ is a frame in
$G$, is \emph{trackable} if the following hold:

\begin{itemize}
\item every doily in $G$ has at least nine vertices;
\item $G$ contains a doily if and only if $F$ is a frame for some minimum
  doily of~$G$;
\item no vertex of $G$ is the centre of a star cutset that, for some $i \in \{1,2,3\}$, separates $\{a_i\}$
from $\{a_1,a_2,a_3\}\setminus \{a_i\}$.
\end{itemize}

The following is easy to prove with brute force enumeration.

\begin{lemma}
  \label{l:cleaning}
  There is an $O(n^{16})$-time algorithm whose input is a graph $G$,
  and whose output is a doily of $G$ with at most eight vertices if 
there is one, and otherwise $k \leq n^{12}$ pairs $(G_1, F_1), \dots, (G_k,
  F_k)$ such that:
  \begin{itemize}
  \item for all $i=1, \dots, k$, $G_i$ is an induced subgraph of $G$
    and $F_i$ is a frame in $G_i$;
  \item $G$ contains a doily if and only
    if for some $i\in \{1, \dots, k\}$, $(G_i, F_i)$ is trackable and
    $G_i$ contains a doily.
  \end{itemize}
\end{lemma}

\begin{proof}
  First, check all subsets of cardinality at most eight from $V(G)$, and stop if one of
  them induces a doily.  This takes time $O(n^8)$.  Now, generate all
  twelve-tuples $F_1, \dots, F_k$ from $V(G)^{12}$.  For each $F_i$ in 
turn, let
$$F_i = (a_1, a_2,
  a_3, b_1, b_2, b_3,  a'_1, a''_1, a'_2, a''_2, a'_3, a''_3),$$
and let $X$ be the set
$$\{a_1, a_2,
  a_3, b_1, b_2, b_3,  a'_1, a''_1, a'_2, a''_2, a'_3, a''_3\}.$$
First we check whether
\begin{itemize}
\item $b_1,b_2,b_3$ are distinct, and different from all of $a_1, a_2,
  a_3, a'_1, a''_1, a'_2, a''_2, a'_3, a''_3$
\item for $i = 1,2,3$, $b_i$ is adjacent to $a_i$, and has no other neighbour in $X$
\item $a_1, a_2, a_3$ are distinct, 
\item for $i = 1,2,3$, $a_i', a_i''$ are distinct neighbours of $a_i$, and $a_i$ has no neighbours
in $X$ except $b_i,a_i', a_i''$
\item for $i = 1,2,3$, if $a_{i-1}, a_{i+1}$ are adjacent then $a_{i-1}' = a_{i+1}$ and $a_{i+1}'' = a_{i-1}$,
and if $a_{i-1}, a_{i+1}$ are non-adjacent then $a_{i-1}',a_{i+1}'' $ are different from all of
$$a_{i-1},a_{i-1}'', a_i', a_i, a_i'', a_{i+1}', a_{i+1}.$$
\end{itemize} 
If one of these is false, go to the next twelve-tuple.  Otherwise build $G_i$ as
  follows.  Initally set $G_i = G$. For $i= 1, 2, 3$, delete from $G_i$ all neighbours of $b_i$ except
  $a_i$, and delete all neighbours of $a_i$ except $b_i, a'_i, a''_i$.

  Now, go through the following loop.  While there exists a vertex $v$
  that is the centre of a star cutset of $G_i$ that, for some $j\in \{1,2,3\}$, separates $\{a_j\}$
  from $\{a_1, a_2, a_3\}\setminus\{a_j\}$, put $G_i \leftarrow G_i\sm v$.  Note that this
  loop can be performed in time $O(n^2|E(G)|)$, because for each $v$,
  computing the connected components of $G_i \sm (N(v) \sm \{a_1, a_2,
  a_3\})$ decides whether $v$ is the centre of a star cutset that
  separates $\{a_j\}$ from $\{a_1,a_2, a_3\}\setminus\{a_j\}$.  If some vertex of $F_i$
  is erased during the loop, go to the next
  12-tuple.  Clearly, after the loop, $F_i$ is a frame in $G_i$ and no
  vertex of $G_i$ is the centre of a star cutset that separates some $\{a_j\}$
  from $\{a_1,a_2, a_3\}\setminus\{a_j\}$. 

  If $G$ contains no doily, then clearly the same is true for $G_i$ for all $i \in \{1,
  \dots, k\}$.  Suppose conversely that $G$ contains a doily.  Let $K$ be
  a minimum doily of $G$.  Since no vertex of $K$ is the centre of a
  star cutset that separates vertices of the cycle of $K$, it follows
  that for some twelve-tuple $F_i$ made of vertices of $K$, a pair $(G_i,
  F_i)$ is generated, such that $G_i$ contains $K$ and $F_i$ is a
  frame for $K$.  So, $(G_i, F_i)$ is a trackable pair and $G_i$
  contains a doily.
\end{proof}

The following is less obvious.

\begin{lemma}
  \label{l:detectDoily}
  There is an $O(n^2)$-time algorithm, whose input is a pair $(G, F)$,
  where $G$ is a graph and $F$ a frame in $G$, and whose output is an
  induced subgraph $K$ of $G$, such that if $(G, F)$ is a trackable
  pair and $G$ contains a doily, then $K$ is a doily.
\end{lemma}

The proof of \ref{l:detectDoily} is postponed to the next two sections:
in Section~\ref{s:cleaning}, we show that when $(G, F)$ is a trackable
pair, every vertex attaches ``locally'' to every minimum
doily of $G$ with frame $F$ (this will be defined formally); and in
Section~\ref{sec:spd}, we take advantage of this to
prove~\ref{l:detectDoily} with the shortest path detector method.
Assuming all this, we can now prove our main result.

\subsection*{Proof of \ref{th:main}}

Here is an algorithm.  Step~1: run the algorithm from \ref{l:cleaning}
for $G$.  If a doily on at most eight vertices is found, then stop.  Otherwise,
go to Step~2: run the algorithm from \ref{l:detectDoily} for all pairs
$(G_i, F_i)$ generated in Step~1.  If some doily is found, then
stop.  Otherwise output ``$G$ contains no doily'' and stop.

Let us prove the correctness of this algorithm.  If the algorithm
outputs a doily, then $G$ obviously contains a doily (so the answer is
correct).  Suppose conversely that $G$ contains a doily.  If some doily of
$G$ has fewer than nine vertices, it is detected in Step 1.  Otherwise,
Step~1 provides a trackable pair $(G_i, F_i)$ that contains a doily.
So, in Step~2, when $(G_i, F_i)$ is considered (or possibly before), a
doily is found. This proves~\ref{th:main}.~$\Box$

\section{Cleaning major vertices}
\label{s:cleaning}

When $P$ is a path or cycle, its \emph{length} is the number of edges in $P$.
When $H$ is an induced cycle and $b \notin V(H)$ is a vertex with a
unique neighbour $a\in V(H)$, we say that the vertex $b$ is a
\emph{tuft for $H$ at $a$}.  We say that vertices $b_1,b_2,b_3$ form a {\em tufting} for $H$ if $b_1,b_2,b_3$ are pairwise non-adjacent,
each is a tuft for $H$, and no two of them have the same neighbour in $V(H)$. So, a doily is an induced cycle with a tufting.
A \emph{hole} is an
induced cycle with at least four vertices.  When $K$ is a subgraph of some
graph $G$ and $v$ a vertex of $G$, we define $N_K(v) = N(v) \cap
V(K)$.

When $K$ is a doily of $G$ and $F$ is a frame for~$K$, we use the
following notation and definitions.  Let $F = (b_1, b_2, b_3, a_1,
a_2, a_3,  \dots)$ (we do not need to name the other vertices of
the frame in this section).  We denote by $H_K$ the unique cycle
of~$K$.  So, $H_K$ is the union of three disjoint chordless paths $P_1, P_2, P_3$,
where for $i = 1,2,3$, $P_i$ is the path of $K\setminus a_i$ from $a_{i+1}$ to $a_{i-1}$

In what follows, all computations with indices are
modulo three.  Let us assign an orientation ``clockwise'' to the cycle $H_K$, such that $a_1,a_2,a_3$ are in clockwise order.
For any vertex $v$ of $H_K$, let $v^+$ be the vertex of $H_K$ that follows $v$ in clockwise order, and let $v^-$
be the vertex that precedes $v$. 
A vertex $v\in V(G)\sm V(K)$ is \emph{minor (with respect to $K$)} if for some
$i=1, 2, 3$, $N_K(v)$ is a subset of some subpath of $P_i$ of length at most two.
A vertex $v\in V(G) \sm V(K)$ is \emph{major (with respect
  to $K$)} if it has neighbours in $P_i$ for all $i=1, 2, 3$.  For $i=1,
2, 3$, when $v$ has at least one neighbour in $P_{i}$, we define
$y_{i+1}(v)$ as the neighbour of $v$ in $P_{i}$ that is closest to
$a_{i+1}$ (along $P_{i}$).  Note that $v$ is non-adjacent to $a_{i-1}, a_{i+1}$ because the
latter both have degree three, and so $y_{i+1}(v)$ is an internal vertex
of $P_{i}$.
Similarly,
for $i=1, 2, 3$, when $v$ has at least one neighbour in $P_{i}$, we
define $x_{i-1}(v)$ as the neighbour of $v$ in $P_{i}$ that is closest to
$a_{i-1}$ (along $P_{i}$).
For $i=1, 2, 3$, when $v$
has neighbours in $P_{i-1}$ and $P_{i+1}$, we define $W_i(v)$ to be the path 
$x_{i}(v) \d P_{i+1} \d a_i \d P_{i-1} \d y_{i}(v)$.

Our goal in this section is the following statement.

\begin{lemma}
  \label{l:noMajor}
  If $G$ is a graph and $(G, F)$ is a trackable
  pair, and $K$ is a minimum doily of $G$, with frame $F$,
  then every vertex in $V(G)\setminus V(K)$ is minor with respect to $K$ in $G$.
\end{lemma}

Throughout this section, $G$ is a graph that contains a doily and
$(G, F)$ is a trackable pair, and $K$ is some minimum doily of $G$, with frame $F$.  
The proof
goes through several lemmas.  The idea is that if there are major
vertices for $K$, then one of them is the centre of a
star cutset that separates $a_1$ from $\{a_2, a_3\}$, which
contradicts that $(G, F)$ is trackable.  
Note that from the definition of a
trackable pair, $K$ has at least nine vertices (so $H_K$ is of length at
least six, but possibly some $P_i$ is of length one).

\begin{lemma}
  \label{l:mM}
  If $v\in V(G) \sm V(K)$, then either $v$ is minor or $v$ is major.
  Suppose that $v$ is major.  Then for all $i\in\{1, 2, 3\}$, 
$x_i(v)$, $x_i(v)^-$, $x_i(v)^{--}$ are
  internal vertices of $P_{i+1}$ and are all adjacent to $v$; and 
$y_i(v)$, $y_i(v)^+$, $y_i(v)^{++}$
  are internal vertices of $P_{i-1}$, and all adjacent to $v$.
\end{lemma}

\begin{proof}
  First we show that if $N_K(v) \subseteq V(P_i)$ for some $i=1, 2, 3$, then $v$ is
  minor.  For if $v$ has neighbours in $P_i$ that are not in a
  three-vertex path of $P_i$, then $P_i\cup \{v\}$ contains a chordless 
  path $P'_i$ from $a_{i-1}$ to $a_{i+1}$, shorter than $P_i$.
  Replacing $P_i$ by $P'_i$ in $K$, we obtain a doily smaller than $K$,
  a contradiction.

  Hence, from the symmetry,  we may assume that $v$ has neighbours in $P_1$ and $P_{2}$.
  Suppose that $v$ has no neighbours in $P_{3}$. Then
  $$H = v \d x_{1}(v) \d P_{2} \d a_1 \d P_{3} \d a_{2} \d P_1 \d y_{2}(v) \d v$$
 is a hole and $b_1, b_{2}$ are tufts for $H$.  We
  claim that $H$ is smaller than $H_K$.  Suppose not; then $x_{1}(v)$ and
  $y_{2}(v)$ are neighbours of $a_{3}$, and so $H' = a_{3} \d
  x_{1}(v) \d v \d y_{2}(v) \d a_{3}$ is a hole of length four.  Since $H_K$ has length at least six (because every
doily has at least nine vertices), it follows that $y_{2}(v)^{-}$,
  $x_{1}(v)^{+}$ are non-adjacent. Hence
  $b_{3}$, $y_{2}(v)^{-}$ and $x_{1}(v)^{+}$ form a tufting for
  $H'$, giving a seven-vertex doily, a contradiction.  This
  proves our claim that $H$ is smaller than $H_K$.  So, since $b_1, b_{2}$ are tufts for $H$ and
  $H$ is smaller than $H_K$, it cannot be that $H$ has a third tuft at some vertex different from $a_1,a_{2}$.
  In particular, $v$ is adjacent to $y_{2}(v)^+$ (for
  otherwise $y_{2}(v)^+$ is a tuft at $y_{2}(v)$).  Symmetrically,
  $v$ is adjacent to $x_{1}(v)^-$, and, in particular, $x_1(v)^-$ is different from $a_{3}$, and so $x_{1}(v)$
  is non-adjacent to $y_{2}(v)^{++}$.  Now, $v$ is non-adjacent to
  $y_{2}(v)^{++}$ (for otherwise $y_{2}(v)^{++}$ is a tuft for $H$ at
  $v$).  Hence, $v \d y_{2}(v) \d y_{2}(v)^{+} \d v$ is a triangle
  for which $x_{1}(v)$, $y_{2}(v)^{-}$ and $y_{2}(v)^{++}$ form a tufting, giving a six-vertex doily, a contradiction.
  Thus we have proved that $v$ has
  neighbours in $P_{3}$, so $v$ is major.

  Let us now prove the second statement of the theorem, and we may assume that $i = 1$.  Assume that $v$ is major; then the hole
  $H$ formed by $W_1(v)$ and $v$ is smaller than $H_K$, and $b_{1},
  y_{2}(v)$ are tufts for $H$.  So, $H$ cannot have a third tuft at some vertex different from $a_1,v$.
  Now $x_1(v)^-$ is non-adjacent to $y_{2}(v)$ because
  otherwise $v \d x_1(v) \d a_{3} \d y_{2}(v) \d v$ is a hole for
  which $x_1(v)^{+}$, $y_{2}(v)^{-}$ and $b_{3}$ form a tufting, 
giving a seven-vertex doily, a contradiction.  Hence, $v$ is
  adjacent to $x_1(v)^-$ for otherwise $x_1(v)^-$ would be a third tuft
  for $H$ at $x_1(v)$.  In particular, $x_1(v)^-\ne a_{3}$. Symmetrically, $v$ is adjacent to
  $y_{2}(v)^+$, and $y_{2}(v)^+\ne a_{3}$. Consequently $x_1(v)^{--}$ is non-adjacent to $y_{2}(v)$.
  Now, $v$ is adjacent to $x_1(v)^{--}$, for otherwise, $v \d x_1(v) \d
  x_1(v)^- \d v$ is a triangle for which $y_{2}(v)$, $x_1(v)^{+}$ and
  $x_1(v)^{--}$ form a tufting.  It follows that $x_1(v)^{--}$ is an internal vertex of $P_{2}$. Thus we have
proved that $v$ is adjacent
  to $x_1(v)$, $x_1(v)^-$ and $x_1(v)^{--}$.  Symmetrically, $y_1(v)$,
  $y_1(v)^+$, $y_1(v)^{++}$ are all internal vertices of $P_{i-1}$, and $v$ is adjacent to them all.
\end{proof}

Let $u$ and $v$ be two major vertices and $i \in \{1, 2,
3\}$.  We say that $u$ and $v$ \emph{disagree at $i$} when $x_{i}(u)$,
$x_{i}(v)$, $y_{i}(u)$, $y_{i}(v)$ are pairwise distinct and appear
along the path $a_{i-1} \d P_{i+1} \d a_i \d P_{i-1} \d a_{i+1}$ in one of the
following orders: 
\begin{itemize}
\item $x_{i}(u)$, $x_{i}(v)$, $y_{i}(u)$, $y_{i}(v)$ or
\item $x_{i}(v)$, $x_{i}(u)$, $y_{i}(v)$, $y_{i}(u)$.  
\end{itemize}
Note that, if $u,v$ are non-adjacent, then $u$ and $v$
disagree at $i$ if and only if there exists an induced path from $u$
to $v$ that goes through $a_i$ and whose interior is in $a_{i-1} \d
P_{i+1} \d a_i \d P_{i-1} \d a_{i+1}$.  Moreover, this induced path is
unique and we denote it by $W_i(u, v)$.

A vertex $z \in V(H_K)$ is a \emph{tie at $i$ for $u$ and $v$} if $i
\in \{1, 2, 3\}$ and either $z = x_i(u) = x_i(v)$ or $z = y_i(u) =
y_i(v)$.  

We say that \emph{$u$ beats $v$ at $i$} when $x_{i}(u)$, $x_{i}(v)$,
$y_{i}(u)$, $y_{i}(v)$ are pairwise distinct and appear along $a_{i-1}
\d P_{i+1} \d a_i \d P_{i-1} \d a_{i+1}$ in the following order:
$x_{i}(v)$, $x_{i}(u)$, $y_{i}(u)$, $y_{i}(v)$.

It is clear that when $u$ and $v$ are major vertices, then for each
$i=1, 2, 3$ exactly one of the following holds: $u$ and $v$ disagree
at $i$; or there is a tie for $u$ and $v$ at $i$; or $u$ beats $v$ at $i$;
or $v$ beats $u$ at $i$.

\begin{lemma}
  \label{l:dis}
  If $u$ and $v$ are two non-adjacent major vertices, then they
  disagree at at most one $i \in \{1, 2, 3\}$.
\end{lemma}

\begin{proof}
  Suppose that $u$ and $v$ disagree at $1$ and $2$ say.
 Then $W_2(u, v) \cup W_{1}(u, v)$ is a hole
  smaller than $H_K$ for which $b_1, b_{2}$ and one of $x_{3}(u), x_{3}(v)$ are
  non-adjacent tufts unless $x_{3}(u) = x_{3}(v)$.  In this last
  case, $W_{1}(u, v) \cup \{x_{3}(u)\}$ is a hole for which
  $b_{1}$ and the second and penultimate vertices of $W_2(u, v)$ are
  non-adjacent tufts, a contradiction to the minimality of~$K$.
\end{proof}

\begin{lemma}
  \label{l:2ties}
  Let $u$ and $v$ be two non-adjacent major vertices.  If there is a
  tie for $u, v$ at distinct $i, j \in \{1, 2, 3\}$, then $N_K(u) =
  N_K(v)$.
\end{lemma}

\begin{proof}
If $z_i$ is a tie for $u,v$ at some $i\in \{1,2,3\}$, let  $z_i' = z_i^+$ if $z_i\in V(P_{i+1})$, and
$z_i' = z_i^-$ if $z_i\in V(P_{i-1})$.
Suppose that for some $i$, $z_i$ and $z_{i+1}$ are ties for $u, v$ at $i$ and
  $i+1$, but that there is no tie for $u, v$ at $i-1$.  So $H =
  u \d z \d v \d z' \d u$ is a hole. 
Then  $z_i'$ and $z_{i+1}'$ are tufts for
  $H$.  Since there is no tie for $u, v$ at $i-1$, from the symmetry we may assume that
  $x_{i-1}(u)$ is closer to $a_{i-1}$ than $x_{i-1}(v)$.  Hence,
  $x_{i-1}(u)$ is a third tuft for $H$ at $u$, forming a seven-vertex doily, a contradiction.

Thus there exist $z_1,z_2,z_3$ such that for $i = 1,2,3$, $z_i$ is a tie for $u, v$ at $i$.
Suppose that $N_K(u) \ne N_K(v)$, and let
$$w\in (N_K(u)\cup N_K(v))\setminus (N_K(u)\cap N_K(v)).$$
By~\ref{l:mM}, $w$ is adjacent to at most one of $a_1,a_2,a_3$, and so
we may assume that $w$ is non-adjacent to $a_1,a_2$. 
  So the subgraph $H_S$ induced on $\{z_1,z_2,u, v\}$ is
  a hole, and $z_1', z_2',w$ form a tufting
  for $H_S$, giving a seven-vertex doily, a contradiction. 
\end{proof}

\begin{lemma}
  \label{l:1tie}
  Let $u$ and $v$ be two non-adjacent major vertices.  If there is a
  tie for $u, v$ at $i \in \{1, 2, 3\}$, then either $N_K(u) \sm
  N_K(a_i) \subseteq N_K(v)$ or $N_K(v) \sm N_K(a_i) \subseteq
  N_K(u)$.  In particular $u$ and $v$ do not disagree at any $j\in
  \{1, 2, 3\}$.
\end{lemma}

\begin{proof}
  By \ref{l:2ties}, we may assume that there is a tie  for
  $u, v$ at $1$, but not at $2$ or $3$.  From the symmetry, we may assume that
  $y_{1}(u) = y_{1}(v)$.  By~\ref{l:mM},
  $u$ and $v$ are adjacent to $y_{1}(u)^+$ and $y_{1}(u)^{++}$.  So,
  $H = u \d y_{1}(u) \d v \d y_{1}(u)^{++} \d u$ is a hole and
  $y_{1}(u)^{-}$ is a tuft for $H$ at $y_{1}(u)$.

Since there is no tie for $u,v$ at $2$, we may assume that
$y_{2}(u)$ is strictly between $a_{2}$ and $y_{2}(v)$ on $P_1$. We claim that
$N_K(v) \sm N_K(a_1) \subseteq N_K(u)$. For suppose that there exists $w\in N_K(v) \sm N_K(a_1)$
such that $w\notin N_K(u)$. By \ref{l:mM}, $w$ is non-adjacent to $y_{2}(u)$, and so
$w, y_{1}(u)^{-}, y_{2}(u)$ are pairwise non-adjacent (in particular, $w,y_{1}(u)^{-}$ are non-adjacent since
$w\notin N_K(a_1)$). But 
$w, y_{1}(u)^{-}, y_{2}(u)$ are not a tufting for $H$, since they would
give a seven-vertex doily; and so $w$ is adjacent to $y_1(u)^{++}$.
In particular, $w\ne x_{3}(v)$, and so $x_{3}(v)$ is adjacent to $u$.
Hence $u \d y_{1}(u) \d v \d x_{3}(v) \d u$ is a hole and
$w,y_{1}(u)^{-},y_{2}(u)$ form a tufting, giving
a seven-vertex doily, a contradiction.
\end{proof}

\begin{lemma}
  \label{l:dis1}
  Let $u$ and $v$ be two non-adjacent major vertices that disagree at some $i\in \{1,2,3\}$. Then either $N_K(u) \setminus N_K(v)$
is a clique (and hence has one or two members) or $N_K(v) \setminus N_K(u)$ is a clique.
\end{lemma}
\begin{proof}
Let $i = 1$ say.
By~\ref{l:1tie}, there is no
  tie for $u$ and $v$ at $1$, $2$ or $3$, and by~\ref{l:dis}, $u$ and $v$ do not
  disagree at $2$ or at $3$.  So, we may assume up to symmetry
  that $v$ beats $u$ at $2$.  Let $u',v'$ be the neighbours of $u,v$
respectively in $W_1(u,v)$. Now at most one vertex in 
$N_K(u) \setminus (N_K(v)\cup \{u'\})$ is adjacent to $u'$; so, since we may 
assume that $N_K(u) \setminus N_K(v)$ is not a clique, it follows that
there exists $w\in N_K(u) \setminus (N_K(v)\cup \{u'\})$ nonadjacent to $u'$. 
Hence $w$ is not in $W_1(u,v)$. By \ref{l:mM}, $v$ is adjacent to 
the neighbour of $v'$ in $K$
not in $W_1(u,v)$; so $w$ is nonadjacent to $v'$, and hence $w$ has no 
neighbour in $W_1(u,v)$ except $u$.

Suppose that $w\notin V(P_1)$. 
If $y_{2}(u)$ is adjacent to $y_{2}(v)$ and hence $v$ is adjacent 
to $y_{2}(u)^+$ by \ref{l:mM}, let $Q$ be the path $u\d y_{2}(u)^+ \d v$, 
and otherwise let $Q$ be the induced
  path from $u$ to $v$ with interior in $y_{2}(u) \d P_1 \d
  a_{2}$.
In either case, since $v$ beats $u$ at $2$, it follows that $y_{2}(v)$ has no 
neighbour in $V(Q)$ except $v$. But then $Q\cup W_1(u,v)$ is a hole and
$b_1,y_2(v), w$ are three tufts for it, contrary to the minimality of $K$.

This proves that $w\in V(P_1)$. But $w, y_2(v)$ are nonadjacent since 
$v$ is adjacent to $y_2(v)^+$. Let $Q'$ be the induced path between $u,v$ with 
interior in $x_{2}(u) \d P_{3} \d a_{2}$. Then $Q'\cup W_1(u, v)$ is a hole and
$b_1,y_2(v), w$ are three tufts for it, contrary to the minimality of $K$.
\end{proof}

We denote the length of a path $P$ by $|P|$.
\begin{lemma}
  \label{l:1jump}
  Let $v$ be a major vertex such that
  $|W_1(v)|$ is minimum and, subject to that, such that
  $|W_{2}(v)|+|W_{3}(v)|$ is minimum.  If $u$ is a major vertex
  non-adjacent to $v$ that has
  neighbours in the interior of $W_1(v)$, then $N_K(u) \setminus N_K(v)$ is a clique,
and $v$ beats $u$ at $2$ and $3$.
\end{lemma}

\begin{proof}
  From the minimality of $|W_1(v)|$ and the fact that $u$ has
  neighbours in the interior of $W_1(v)$, we know that $u$ and $v$
  disagree at $1$.  Hence, by \ref{l:1tie}, there is no tie between $u,v$ at $1$, $2$ or $3$.
So, by~\ref{l:dis1}, we may assume that $N_K(v)\setminus N_K(u)$ is a clique.
  It follows that $u$ beats $v$ at $2$ and
  $3$, so $y_{2}(u)$ is not adjacent to $v$.  We may assume that
  $x_1(u)$, $x_1(v)$, $y_1(u)$, $y_1(v)$ appear in
  this order along $a_{3} \d P_{2} \d a_1 \d P_{3} \d
  a_{2}$.

  If $y_{1}(u)^+ = y_{1}(v)$ then $u \d y_{1}(u) \d y_{1}(v) \d u$ is
  a triangle for which $y_{1}(u)^{-}, v, y_{2}(u)$ form a tufting,
  giving a six-vertex doily,  a contradiction.  So, $y_{1}(u)^+
  \neq y_{1}(v)$.  From the minimality of $|W_1(v)|$, it follows that
  $x_{1}(u) \d P_{2} \d x_{1}(v)$ has length at least two.  If it has
  length exactly two, then $|W_1(u)|=|W_1(v)|$ and, since $u$ beats $v$
  at $2$ and $3$, there is a contradiction to the optimality
  of~$v$.  So, the length of $x_{1}(u) \d P_{2} \d x_{1}(v)$ is
  at least three.  Now let $Q$ be an induced path from $u$ to $v$
  whose interior is in $y_{2}(v) \d P_1 \d a_{2}$.  Then $Q
  \cup W_1(u, v)$ is a hole for which $x_{3}(u)$, $x_{1}(v)^{--}$ and
  $b_1$ form a tufting, a contradiction to the minimality
  of~$K$.
\end{proof}

\noindent{\bf Proof of \ref{l:noMajor}.\ \ }
Let $(G,F)$ be a trackable pair.
By \ref{l:mM}, we only need
to prove that no minimum doily of $G$ has frame $F$ and has a major vertex.
To do so, we prove that if some minimum doily has frame 
$F$ and has a major vertex, then there is a star cutset that separates $\{a_1\}$
from $\{a_2, a_3\}$, which is a contradiction to the trackability of
$(G, F)$.  

We assume therefore that there is a minimum doily $K$ of $G$ that has frame $F$ and has a major vertex $v$; 
and let us choose $K,v$ as follows.
\begin{itemize}
\item[{\bf (i)}] Among all such choices of $K,v$, let us choose $K,v$ such that $|W_1(v)|$ is minimum.
\item[{\bf (ii)}] Among all such choices of $K,v$ satisfying condition (i), let us choose $K,v$
such that $|W_{2}(v)|+|W_{3}(v)|$ is minimum.
\item[{\bf (iii)}] Among all choices of $K,v$ satisfying (i) and (ii) above, since  $v$  is not the centre of a star cutset that separates
$\{a_1\}$ from $\{a_2, a_3\}$, there is a path $P= p_1 \c p_k$ disjoint from $K$ such that 
$k\geq 1$, $p_1$ has neighbours in the interior of $W_1(v)$, $p_k$
has neighbours in $V(K) \sm (W_1(v) \cup N_K(v))$, 
and no vertex of $P$ is a neighbour of $v$. Let us choose $K,v$ such that
such a path $P = p_1 \c p_k$ exists with $k$ minimum.
\end{itemize}

The first two conditions will later be referred to as the 
\emph{optimality} of $v$, and the third the
\emph{minimality} of $P$.  We now look for a contradiction.  (This will
prove~\ref{l:noMajor}.)
\\
\\
(1) {\em $P$ is induced, and no vertex of $P \sm p_1$
has a neighbour in the interior of $W_1(v)$, and no vertex of $P \sm p_k$ has a neighbour
in $V(K) \sm (W_1(v) \cup N_K(v))$.}
\\
\\
This is immediate from the minimality of $P$.
\\
\\
(2) {\em $k\ge 2$; and if $p_1$ is major then $N_K(p_1)\setminus N_K(v)$ is a clique and $v$ beats $p_1$ at~$2$ and~$3$.}
\\
\\
The second assertion follows from \ref{l:1jump}, so it remains to prove that $k\ge 2$. 
Suppose that $k = 1$. Then $p_1$ has a neighbour in $V(K) \sm (W_1(v) \cup N(v))$; let 
$Q$ be a minimal subpath of $H_K\setminus a_1$ containing a neighbour of $p_1$ in the interior of 
$W_1(v)$ and a neighbour of $p_1$ in $K \sm (W_1(v) \cup N(v))$. Then $Q$ has at least three internal vertices by \ref{l:mM}. Consequently $p_1$
is major, contrary to the second assertion. This proves (2).
\\
\\
(3) {\em $p_2$ is adjacent to both or neither of $x_1(v), y_1(v)$.}
\\
\\
For suppose it is adjacent to exactly one, say $y_1(v)$. Let $t\in \{y_2(v), x_3(v)\}$.
Now, $W_1(v)$ and $v$ form a hole for
  which $b_1,p_2$ and $t$ are tufts; so $p_2$ is adjacent to $t$, and so $p_2$ is adjacent to both $y_2(v), x_3(v)$. 
In particular, $p_2$ is major.
Since $x_1(v)$ is non-adjacent to $p_2$, and there is a tie for $v,p_2$ at $1$, \ref{l:2ties} implies that 
there is no tie for $v,p_2$ at $2$, and so $v$ is non-adjacent to $y_2(p_2)$. 
  Then $v\d y_1(v)\d p_2\d x_3(v)\d v$ is a hole for
  which $y_1(v)^{-}$, $x_1(v)$ and $y_2(p_2)$ form a tufting
  (note that $y_1(v)^{-}, x_1(v)$ are non-adjacent because $p_1$ has a neighbour in the interior of $W_1(v)$ different from $a_1$) 
contrary to the minimality of $K$. This proves (3).

\bigskip

From the optimality of $v$, not both the paths $a_1 \d P_3 \d y_1(v)$ and  $a_1\d P_2\d x_1(v)$ contain neighbours of $p_1$. Thus we may
assume from the symmetry that $p_1$ has no neighbours in $a_1\d P_2\d x_1(v)$.
It follows that $p_1$ has a neighbour in the interior of $a_1 \d P_3 \d y_1(v)$, and so this path has length at least two.
\\
\\
(4) {\em If $p_2$ is adjacent to both $x_1(v), y_1(v)$, then $p_2$ is non-adjacent to both of $y_2(v), x_3(v)$.}
\\
\\
For suppose that $p_2$ is adjacent to one of
$y_2(v), x_3(v)$, say $t$. Let $t' = y_2(v)^-$ if $t = y_2(v)$, and $t' = x_3(v)^+$ if $t = x_3(v)$.
Then $v\d x_1(v)\d p_2\d t\d v$ is a hole, with a tufting
$x_1(v)^+, t',p_1$, a contradiction. 
This proves (4).
\\
\\
(5) {\em If $p_2$ is adjacent to both $x_1(v), y_1(v)$, then $W_1(v)$ has length three.}
\\
\\
For the hole $v\d y_1(v) \d p_2 \d x_1(v)\d v$
has three tufts $x_1(v)^{+}$, $y_1(v)^{-}$, and $y_2(v)$. Since these do not form a tufting,
we deduce that $x_1(v)^{+}$, $y_1(v)^{-}$ are adjacent,
and so $W_1(v)$ has length three. This proves (5).

\bigskip
Thus, in the case that $p_2$ is adjacent to both $x_1(v), y_1(v)$, 
since $p_1$ has a neighbour in the interior of  $a_1 \d P_3 \d y_1(v)$, it follows
that $x_1(v) = a_1^-$, and $a_1^+$
is adjacent to $y_1(v)$. Note also that in this case $p_2$ is major.
\\
\\
(6) {\em If $p_2$ is adjacent to both $x_1(v), y_1(v)$, then $p_1$ is adjacent to $a_1^+$ and to $y_1(v)$, and $p_1$ is major.}
\\
\\
For from the definition of $p_1$ it follows that $p_1$ is adjacent to $a_1^+$. 
Since $a_1,p_1,y_2(v)$ is not a tufting for the hole $v\d x_1(v)\d p_2\d y_1(v)\d v$,
it follows that $p_1$ is adjacent to $y_1(v)$. Let $u = x_1(v)^{--}$.
Since $a_1, x_3(v),p_1$ is not a tufting for the hole $v\d x_1(v)\d p_2\d u\d v$, it follows that
$p_1$ is adjacent to $u$ and hence $p_1$ is major.  
This proves (6).
\\
\\
(7) {\em If $p_2$ is adjacent to both $x_1(v), y_1(v)$, then $p_2$ is non-adjacent to $y_2(p_1), x_3(p_1)$.}
\\
\\
For suppose that $p_2$ is adjacent to one of $y_2(p_1), x_3(p_1)$, say $t$. Let $t' = y_2(p_1)^-$ if $t = y_2(p_1)$, and $t' = x_3(p_1)^+$ if
$t =  x_3(p_1)$. Then the subgraph induced on $\{p_1,p_2,t,a_1^+,x_1(v), t'\}$ is a six-vertex doily, a contradiction. This proves
(7).
\\
\\
(8) {\em If $p_2$ is adjacent to both $x_1(v), y_1(v)$, then $k\ge 3$, and 
$p_3$ is adjacent to $x_1(v)$ and to $y_1(v)^{+}$.}
\\
\\
For since there is a tie for $v,p_2$ at $1$, and $y_2(v)\in N_K(v)\setminus N_K(p_2)$, \ref{l:1tie} implies that
$N_K(p_2) \subseteq N_K(v)$. Consequently $k\ge 3$.
Now $a_1\d x_1(v)\d p_2\d p_1\d a_1^+\d a_1$ is a hole, and $x_1(v)$ is the only neighbour of $v$ in this hole, and
$a_1$ is the only neighbour of $b_1$ in this hole. Since $v,b_1, p_3$ are pairwise non-adjacent, and every doily has at least nine vertices,
it follows that $p_3$ is adjacent to $x_1(v)$. Since $\{p_1, p_2, y_1(v)^{+}, a_1^+, v,p_3\}$ does not induce a six-vertex
doily, it follows that $p_3$ is adjacent to $y_1(v)^{+}$. This proves (8).
\\
\\
(9) {\em $p_2$ is non-adjacent to $x_1(v), y_1(v)$.}
\\
\\
For otherwise by (3), $p_2$ is adjacent to both $x_1(v), y_1(v)$.
Since $\{p_2,p_3,x_1(v), p_1, x_3(v),a_1\}$ does not induce a six-vertex doily, and $p_1, p_2$ are non-adjacent to $x_3(v)$
by (2) and (4), it follows that $p_3$ is non-adjacent to $x_3(v)$. But then $p_1,a_1, x_3(v)$ form a tufting for the hole
$v\d y_1(v)^+\d p_3\d x_1(v)\d v$, a contradiction. This proves (9).
\\
\\
(10) {\em  $p_1$ is major with respect to $K$.}
\\
\\
  For suppose that $p_1$ is minor.  If $p_1$ has a unique neighbour $r$ in
  $W_1(v)$, then $r$ is in the interior of $W_1(v)$, so $W_1(v)$ and
  $v$ form a hole for which $b_1$, $p_1$ and $x_3(v)$ are non-adjacent
  tufts, a contradiction to the minimality of $K$.

  Suppose that $p_1$ has exactly two neighbours in $W_1(v)$, say $q$ and
  $r$, and they are adjacent.  
We may assume that $a_1$, $q$, $r$ and $y_1(v)$ appear in this
  order along $P_3$ (possibly $r=y_1(v)$).  Let $r'$ = $r^+$ if $r\ne y_1(v)$, and $r'= v$ if $r = y_1(v)$.
By (9), $q^-, r',p_2$ form a tufting for the cycle $q \d r \d
  p_1 \d q$, a contradiction.

  It follows that $p_1$ has two non-adjacent neighbours $q, s$ in
  $W_1(G)$, and $N(v) \cap V(K) \subseteq \{q, r, s\}$ where $r$ is
  the common neighbour of $q, s$ in $K$.  Then, we obtain another minimum doily $K'$ of $G$, still with frame $F$, 
by replacing $r$ by $p_1$ in $K$.  The doily $K'$ together
  with $v$ and $p_2 \d P \d p_k$ contradicts the minimality of $P$.
From \ref{l:mM}, this proves (10).
\\
\\
(11) {\em $p_1$ has two non-adjacent neighbours
  in $W_1(v)$.}
\\
\\
For suppose not.
 Since $p_1$ is major, it follows from \ref{l:mM} that $p_1$ has exactly two neighbours
 in $W_1(v)$, namely $y_1(v)$ and $y_1(v)^-= y_1(p_1)$.  But now by (9)
 the triangle $p_1 \d y_1(p_1) \d y_1(v) \d p_1$ has a tufting
 $y_1(p_1)^-$, $p_2$ and $v$. This proves (11).
\\
\\
(12) {\em $p_2$ is adjacent to all of $x_2(v)$, $y_2(v)$, $x_3(v)$, $y_3(v)$, and in particular $p_2$ is major.}
\\
\\
For by (9), $p_2$ is non-adjacent to both $x_1(v)$ and $y_1(v)$.  
Let $s$ be the neighbour of $p_1$ closest to $y_1(v)$ along
  $W_1(v)$.  So, 
$$p_1 \d s \d W_1(v) \d y_1(v) \d v \d x_1(v) \d W_1(v) \d y_1(p_1) \d p_1$$ 
is a hole for which $b_1$ and $p_2$ are
  non-adjacent tufts.  Since $x_2(v)$, $y_2(v)$, $x_3(v)$, $y_3(v)$
  are tufts at $v$, $p_2$ is adjacent to all of them, for
  otherwise one of them would be a third tuft at $p_1$.  In particular, $p_2$
  is major.  This proves (12).
\\
\\
(13) {\em $p_2$ beats both $v$ and $p_1$ at both $2$ and $3$.}
\\
\\
For suppose there is a tie $z$ for $p_2$ and $v$ at $2$.
Let $z' = z^+$ if $z\in V(P_3)$, and $z' = z^-$ if $z\in V(P_1)$.
  Then
$$p_1 \d p_2 \d z \d v \d x_1(v) \d W_1(v) \d y_1(p_1) \d p_1$$ 
is a hole for which $z'$ and $b_1$ are
  non-adjacent tufts.  By \ref{l:2ties}, there is no tie for $v$ and
  $p_2$ at $3$, so $x_3(p_2)$ is non-adjacent to $v$, and therefore non-adjacent to $p_1$ by the minimality of $P$, it 
follows that $x_3(p_2)$
  is a third tuft, a contradiction.  Hence, there is no tie for $v$ and $p_2$
at $2$, and similarly none at $3$.
  Thus $p_2$ beats $v$ at both $2$ and $3$.  Since  $v$
  beats $p_1$ by (2), it follows that $p_2$ beats $p_1$ at $2$ and $3$.  This proves (13).

\bigskip

Now, to finally obtain a contradiction:
\begin{itemize}
\item If
  $x_1(p_1)$ and $x_1(p_2)$ are distinct and appear in this order
  along $a_3 \d P_2 \d a_1$, then 
$$p_1 \d p_2 \d x_1(p_2) \d W_1(p_2)
  \d y_1(p_1) \d p_1$$ 
is a hole for which $b_1$, $y_1(p_1)^{++}$ and
  $x_3(p_2)$ are non-adjacent tufts, a contradiction. 
\item If $x_1(p_1) = x_1(p_2)$ then
  $p_1 \d p_2 \d x_1(p_1) \d p_1$ is a triangle for which
  $x_1(p_1)^{+}$, $y_1(p_1)$ and $x_3(p_2)$ is a tufting, a
  contradiction.

\item Finally, suppose $x_1(p_2)$
  and $x_1(p_1)$ are distinct and appear in this order along $a_3 \d
  P_2 \d a_1$. From (11) and the optimality of $v$, it follows that $x_1(p_1), x_1(v)$ are non-adjacent.
Let $u$ be the neighbour of $v$ in $W_1(p_1)$ closest to $x_1(p_1)$. By \ref{l:mM}, $u, x_1(v)$ are non-adjacent.
Then by (2), $p_1$ is non-adjacent to $y_2(v)$; by (12), $p_2$ is adjacent to $y_2(v)$; 
and by (13), $p_1$ is non-adjacent to $x_3(p_2)$.
Consequently $x_3(p_2), y_1(p_1), x_1(v)$ is a tufting for the hole
$$p_1\d x_1(p_1)\d W_1(p_1)\d u\d v\d y_2(v)\d p_2\d p_1,$$
a contradiction.
\end{itemize}

This proves~\ref{l:noMajor}.

\section{Shortest path detector}
\label{sec:spd}

Our goal in this section is to prove~\ref{l:detectDoily}.  We need the
following lemma.  

\begin{lemma}\label{nojumps}
Let $K$ be a minimum doily in a graph $G$, and let $F$ be a frame for $K$. 
Suppose that $G$ contains no
  major vertex with respect to $K$. With our usual notation, let $1\le i\le 3$, and let
$s,t\in V(H_K)\setminus a_i$. Let $Q$ be a path in $G\setminus a_i$ between $s,t$,
 and let $P$ be the (unique) path of $K\setminus a_i$ between $s,t$. Then $|Q|\ge |P|$, and if equality holds then no vertex
of the interior of $Q$ belongs to or has a neighbour in $V(H_K)\setminus V(P)$.
\end{lemma}
\begin{proof}
We proceed by induction on $|Q|$ (for all minimum doilies with frame $F$), and for $|Q|$ fixed, by induction on $|V(H_K)\setminus V(P)|$. 
We may assume that $i = 1$ from the symmetry, and that $Q$ is an induced path from the first inductive hypothesis.
\\
\\
(1) {\em We may assume that no internal vertex of $Q$ belongs to $K$.}
\\
\\
For suppose that some internal vertex $r$ of $Q$ belongs to $K$, and hence to $V(H_K\setminus a_1)$.
Let $P_1$ be the path in $K\setminus a_1$ between $s,r$, and let $Q_1$ be the subpath of $Q$ between $s,r$.
Define $P_2,Q_2$ between $r,t$ similarly. From the first inductive hypothesis, $|Q_j|\ge |P_j|$
for $j = 1,2$; and since 
$$|Q| = |Q_1|+|Q_2| \ge |P_1|+|P_2|\ge |P|,$$
it follows that $|Q|\ge |P|$, and we may assume that equality holds. Hence $|P_j| = |Q_j|$ for $j = 1,2$, and $|P_1|+|P_2| = |P|$.
The latter implies that $r\in V(P)$.
From the first inductive hypothesis, for $j = 1,2$, no internal vertex of $Q_j$ belongs to or has a neighbour in 
$V(H_K)\setminus V(P_j)$, and in particular, no internal vertex of $Q_j$ belongs to or has a neighbour in $V(H_K)\setminus V(P)$. Moreover, $r$
does not belong to $V(H_K)\setminus V(P)$ (since $r\in V(P)$), and $r$ has no neighbour in $V(H_K)\setminus V(P)$, because it has precisely two neighbours
in $V(H_K)$ and they both belongs to $V(P)$ (one is in $V(P_1)$ and the other in $V(P_2)$). Thus in this case the result holds. This proves (1).

\bigskip

If $Q$ has length at most one the result is clear. If it has length two, then since its internal vertex is minor by hypothesis,
again the result holds. We may therefore assume that $Q$ has length at least three.
Let $u,v$ be the neighbours of $s,t$ in $Q$, respectively. Then $u,v\notin V(K)$.
We may assume that $a_1, s,t$ are in clockwise order in $H_K$.
Now $H_K\setminus a_1$ is a path $R$ say, between $a_1^-$ and $a_1^+$.  For all $p,q\in V(R)$, $R[p,q]$ denotes the subpath of $R$
between $p,q$. For each $w\in V(G)\setminus V(K)$ with a neighbour in $V(R)$, let $x(w)$
be the neighbour of $w$ in $V(R)$ that is closest (in $R$) to $a_1^-$, and let $y(w)$ be the neighbour
closest to $a_1^+$. 
Since $w$ is not major, it follows that none of $a_1,a_2,a_3$ belong to the path $R[x(w),y(w)]$.
\\
\\
(2) {\em We may assume that no vertex of the interior of $Q$
has a neighbour in $V(H_K)\setminus V(P)$; and in particular $s = y(u)$ and $t = x(v)$.}
\\
\\
For suppose that some internal vertex $q$ of $Q$ is adjacent to some $r\in V(H_K)\setminus V(P)$. We may assume that
$r\in V(R[a_1^+, s])\setminus\{s\}$ from the symmetry. Let $Q'$ be the path $t\d Q\d q\d r$; then $|Q'|\le |Q|$. But from the second inductive 
hypothesis, $|Q'|\ge |R[r,t]|> |P|$, and so $|Q|>|P|$ as required. This proves (2).

\bigskip

If $|Q|>|P|$ there is nothing to prove, and if $|P| = |Q|$ then the result holds by (2).
Thus we may assume that $|Q|<|P|$, and we need to prove that this is impossible.
In particular, $P$ has at least five vertices, and so $K$ has at least nine; and from the minimality of $K$, it follows that there is
no doily in $G$ with at most eight vertices. Let $T$ be the path of $H_K$ between $s,t$ that passes through $a_1$, and let $H$ be the hole
$Q\cup T$. Thus $|V(H)|<|V(H_K)|$.
\\
\\
(3) {\em The subpaths $R[x(u),y(u)]$ and $R[x(v),y(v)]$ are disjoint, and one of $a_2,a_3$ (say $a_h$) belongs to the interior of $R[x(u),y(v)]$.}
\\
\\
For 
if $a_2,a_3$ both belong to $V(H)$ then $b_1,b_2,b_3$ are three tufts for $H$, and since $|V(H)|<|V(H_K)|$, this contradicts the minimality of $K$.
Thus we may assume that one of $a_2,a_3$, say $a_h$, belongs to the interior of $P$. 
Since $a_h$ does not belong to $R[x(u),y(u)]$ or to
$R[x(v),y(v)]$ (because $u,v$ are minor), it follows that these subpaths are disjoint, and $a_h$ belongs to the interior of $R[x(u),y(v)]$.
This proves (3).

\bigskip

Let $u', v'$ be the neighbours of $u,v$ in the interior of $Q$, respectively.
We recall that, since $a_h$ has degree three, it has no neighbours in the interior of $Q$.
\\
\\
(4) {\em $x(u),y(u)$ are either equal or adjacent, and so are $x(v), y(v)$.}
\\
\\
For suppose that $x(u), y(u)$ are distinct and non-adjacent. Then they have a common neighbour $r$ in $H_K$. Replacing $r$ by $u$ in $K$
gives another minimum doily $K'$ of $G$, also with frame $F$; 
and $Q\setminus s$ is a path between $u,t$. Now every major vertex for $K'$
is also major for $K$, and so there are no major vertices with respect to $K'$. From the first inductive hypothesis, it follows that the length
of $Q\setminus s$ is at least one more than the length of $R[x(u),t]$.
But then it follows that
$|Q|\ge |P|$, a contradiction. This proves (4).
\\
\\
(5) {\em If $x(u)\ne s$ then $s^{++}$ is adjacent to $u'$ and to no other vertex of $H$, and $u, v$ are non-adjacent. 
Similarly if $y(v)\ne t$
then $t^{--}$ is adjacent to $v'$ and to no other vertex of $H$, and $u,v$ are non-adjacent.}
\\
\\
For suppose that $x(u)\ne s$. By (4), $x(u) = s^+$. From the first
inductive hypothesis, $s^+$ has no neighbours in $Q$ except $s,u$; and $s^{++}$ has no neighbours in $Q$ except possibly $u'$.
Now $u'$ is non-adjacent to $s^-$ by (2), and non-adjacent to $s^+$ as we have seen, and non-adjacent to 
$s$ since $Q$ is induced. Since the subgraph induced on $\{u,s,s^+, u',s^-,s^{++}\}$ is not a six-vertex doily, it follows that $u'$
is adjacent to $s^{++}$. Since $a_h$ does not belong to the path
$R[y(v), x(v)]$, it follows that $v$ is non-adjacent to $s^{++}$, and so $v\ne u'$, and therefore
$u,v$ are non-adjacent. This proves the first statement of (5), and the second 
follows from the symmetry.

\bigskip

Now if $x(u) = s$ and $y(v) = t$ then $s^+, t^-,b_1$ form a tufting for $H$, a contradiction. Thus from 
the symmetry, we may assume that $x(u)\ne s$. By (4) and (5), 
$s^{++}$ is adjacent to $u'$ and to no other vertex of $H$, and $u, v$ are non-adjacent. In particular $a_h\ne s^{++}$ 
(since $a_h$ has no neighbours in the interior of $Q$), and so $a_h$
belongs to the interior of $R[s^{++}, y(v)]$. If $y(v) = t$ then $s^{++}, t^-, b_1$ form a tufting for $H$, a contradiction;
so $y(v) \ne t$. By (4) and (5), $y(v) = t^-$ and $t^{--}$ is adjacent to $v'$ and to no other vertex of $H$. Consequently $a_h$ belongs to the 
interior of $R[s^{++}, t^{--}]$, and in particular $s^{++}, t^{--}$ are non-adjacent. Since $a_h$ does
not belong to the path $R[x(u'), y(u')]$, it
follows that $u'\ne v'$.
We deduce that $s^{++}, t^{--}, b_1$ form a tufting for $H$,
a contradiction. This proves \ref{nojumps}.
\end{proof}

\begin{lemma}
  \label{l:replace}
  Let $G$ be a graph, and let $F$ be a frame for a minimum doily $K$ of $G$,
with the usual notation.
  Suppose that $G$ contains no
  major vertex with respect to $K$.  Let $i\in \{1, 2, 3\}$, and let $Q_i$ be a shortest path
  from $a_{i-1}$ to $a_{i+1}$ in $G\setminus a_i$.  The graph obtained from $K$ by replacing $P_i$ by $Q_i$
  is a minimum doily of $G$, and has frame $F$, and no vertex is major with respect to it.
\end{lemma}

\begin{proof}
From the choice of $Q_i$ it follows that $|Q_i|\le |P_i|$; and so from \ref{nojumps},
equality holds, and no vertex
of the interior of $Q_i$ belongs to or has a neighbour in $V(H_K)\setminus V(P_i)$.
Consequently the graph obtained from $K$ by replacing $P_i$ by $Q_i$
  is a minimum doily of $G$, say $K'$, and has frame $F$. Since no vertex has neighbours in the interiors of both $P_{i-1}, P_{i+1}$ except $a_i$,
it follows that no vertex is major with respect to $K'$.
\end{proof} 

\bigskip

\noindent{\bf Proof of~\ref{l:detectDoily}.\ \ }
  Suppose that we are given a pair $(G, F)$ where 
$$F = (b_1, b_2, b_3,
a_1, a_2, a_3, a'_1, a_2', a_3', a_1'', a_2'', a_3'')$$
  is a frame in $G$.  Here is an algorithm:

\begin{itemize}
\item For $i=1, 2, 3$, compute a shortest path $Q_i$ from $a'_{i-1}$
  to $a''_{i+1}$ (if such a path does not exist, let $Q_i$ be the null graph).
\item Output the subgraph of $G$ induced on $a_1$, $a_2$, $a_3$, $b_1$, $b_2$, $b_3$ and the
  vertices of $Q_1$, $Q_2$ and $Q_3$.
\end{itemize}

This algorithm obviously outputs an induced subgraph $K$ of $G$.  It
remains to prove that if $(G, F)$ is a trackable pair and $G$ contains
a doily, then $K$ is a doily.  Assume therefore that $(G, F)$ is a trackable pair and $G$ contains
a doily. Consequently, there is a minimum doily $K'$ of $G$ such that $F$ is a frame for $K'$.
By~\ref{l:noMajor}, 
$G$ contains no major vertex with respect to $K'$. By
applying Lemma~\ref{l:replace} three times, we see that $K$ is a doily.

\end{document}